\documentclass[11pt]{article}
\usepackage[utf8]{inputenc}

\usepackage{amsmath, amssymb, amsthm, amsfonts}
\usepackage{mathtools}
\usepackage{thm-restate}

\usepackage[dvipsnames,svgnames,x11names]{xcolor}
\definecolor{ForestGreen}{rgb}{0.1333,0.5451,0.1333}
\usepackage[colorlinks=true,linkcolor=ForestGreen,citecolor=blue]{hyperref}

\usepackage[margin=1in]{geometry}
\usepackage{graphics}
\usepackage{pifont}
\usepackage{tikz}
\usepackage{bbm}
\usepackage[T1]{fontenc}

\usetikzlibrary{arrows.meta}
\usepackage{environ}
\usepackage{framed}
\usepackage{url}
\usepackage[linesnumbered,ruled,vlined]{algorithm2e}
\usepackage[noend]{algpseudocode}
\usepackage[noabbrev,capitalize,nameinlink]{cleveref}
\crefname{equation}{}{}
\usepackage[labelfont=bf]{caption}
\usepackage{cite}
\usepackage{framed}
\usepackage[framemethod=tikz]{mdframed}
\usepackage{appendix}
\usepackage{graphicx}
\usepackage[textsize=tiny]{todonotes}
\usepackage{tcolorbox}
\allowdisplaybreaks[1]

\newcommand\remove[1]{}

\newtheorem{lemma}{Lemma}[section]
\newtheorem{theorem}{Theorem}
\newtheorem*{lemma*}{Lemma}

\newtheorem*{corollary*}{Corollary}

\newtheorem*{remark}{Remark}

\theoremstyle{definition}

\newtheorem*{theorem*}{Theorem}
\newtheorem{definition}[lemma]{Definition}

\newtheorem*{rem*}{Remark}

\newcommand{\eps}{\varepsilon}
\newcommand{\R}{\mathbb{R}}

\renewcommand{\O}{\widetilde{O}}

\crefname{algocf}{Algorithm}{Algorithms}
\crefname{claim}{Claim}{Claims}

\renewcommand{\bar}{\overline}
\renewcommand{\hat}{\widehat}

\newcommand{\wt}{\widetilde}

\renewcommand{\bar}{\overline}
\newcommand{\diag}{\mathsf{diag}}

\newcommand{\g}{\nabla}
\newcommand{\hphi}{\hat{\phi}}
\newcommand{\aL}{\approx L}
\newcommand{\cD}{\mathcal{D}}
\newcommand{\sssp}{\mathsf{SSSP}}
\newcommand{\beps}{\bar{\eps}}
\newcommand{\assign}{\leftarrow}

\begin{document}

\title{Incremental Shortest Paths in Almost Linear Time via a \\ Modified Interior Point Method}

\author{Yang P. Liu \\ Carnegie Mellon University \\ yangl7@andrew.cmu.edu}

\clearpage\maketitle

\begin{abstract}
We give an algorithm that takes a directed graph $G$ undergoing $m$ edge insertions with lengths in $[1, W]$, and maintains $(1+\epsilon)$-approximate shortest path distances from a fixed source $s$ to all other vertices. The algorithm is deterministic and runs in total time $m^{1+o(1)}\log W$, for any $\eps > \exp(-(\log m)^{0.99})$. This is achieved by designing a nonstandard interior point method to crudely detect when the distances from $s$ other vertices $v$ have decreased by a $(1+\epsilon)$ factor, and implementing it using the deterministic min-ratio cycle data structure of [Chen-Kyng-Liu-Meierhans-Probst, STOC 2024].
\end{abstract}

\tableofcontents

\newpage

\section{Introduction}
\label{sec:intro}

In this paper we study the shortest path problem in dynamic directed graphs with positive edge lengths.
Specifically, we focus on \emph{incremental} graphs, those which only undergo edge insertions.
Our goal is to design an algorithm which maintains $(1+\eps)$-approximate shortest paths from a source vertex $s$ to all other vertices -- this is known as the SSSP (single-source shortest path) problem.

Over the last several decades, there has been a vast literature on SSSP in incremental as well as \emph{decremental} graphs (those which undergo only edge deletions) \cite{ES81,DHZ00,RZ11,HKN14,HKN15,HKN16,Bernstein16,KL19,KL20,PVW20,PW20,CZ21,KMP22,GK25}. One motivation is that in both the directed and undirected settings, (decremental) shortest path data structures can often be combined with multiplicative weights updates to give algorithms for maximum flow, multicommodity flow, and several other problems \cite{GK98,M10,BPS20,BPS21,CK24a,CK24b}. In particular, approximate shortest paths suffice for these applications, and additionally, there are conditional lower bounds against exact shortest paths in both incremental or decremental graphs \cite{AW14,HKNS15,AHRVW19,JX22,SVXY25}. However, until recently, almost-linear time algorithms for most incremental or decremental problems in dynamic directed graphs were out of reach.

Building off the almost-linear time mincost flow algorithm of \cite{CKLPPS22}, the works \cite{BLS23,CKLMP24,BCK+24} gave a framework based on a dynamic interior point method (IPM) to solve dynamic minimum cost flow in incremental and decremental graphs. Thus, they achieved almost-linear-time algorithms for approximate $s$-$t$ mincost flow, which in turn give algorithms for single-pair shortest paths, maxflow, cycle detection, and more. However, the SSSP problem cannot be captured by a single demand, and thus these works did not have any direct implications for SSSP.

In this work, we give an almost linear time approximate SSSP data structure in incremental directed graphs by using several modifications of the dynamic IPM driving previous works.
\begin{theorem}
\label{thm:main}
Let $m \ge 1$ and $\eps \ge \exp(-(\log m)^{0.99})$. Let $G$ be an incremental graph undergoing $m$ edge insertions with integral edge lengths in the range $[1, W]$ with a source $s \in V$. There is a deterministic data structure that explicitly maintains distance estimates $\wt{d}: V \to \R_{\ge0}$
\[ d_G(s, v) \le \wt{d}(v) \le (1+\eps)d_G(s, v), \]
and the ability to report approximate shortest paths $\pi_{s,v}$ from $s$ to $v$ in time $O(|\pi_{s,v}|)$. The total running time of the data structure is $m^{1+o(1)} \log W$.
\end{theorem}
The previous best runtime was $\O(m^{3/2} \log W)$ for deterministic algorithms and $\O(m^{4/3} \log W)$ for randomized algorithms (against an adaptive adversary) \cite{KMP22}. The latter bound is the best known even among algorithms against oblivious adversaries. \cref{thm:main} also gives a deterministic APSP data structure in incremental graphs with running time $m^{1+o(1)}n$, which improves over the recent work of \cite{GK25} (the work of \cite{Bernstein16} gave randomized algorithms for incremental and decremental APSP against oblivious adversaries with near-optimal runtimes). We refer the reader to \cite{CKLMP24,BCK+24} for additional references relating to other problems in partially dynamic graphs, as well as conditional hardness for these problems.

It is worth mentioning the recent work \cite{MMNNS25} which studied the incremental SSSP problem \emph{with predictions} -- also see the works \cite{HSSY24,BFNP25} which studied other dynamic graph problems with predictions. The work \cite{MMNNS25} also studied the offline incremental SSSP problem (without predictions) and gave a near-optimal $\O(m/\eps)$ time algorithm. This was also achieved in the later work \cite{GK25}. 

\paragraph{Preliminaries.} We sometimes say that a graph $G$ is \emph{incremental} to mean that it represents a sequence of graphs $G^{(0)}, G^{(1)}, \dots$, where each graph has an extra edge compared to the previous one. We let $B \in \R^{E \times V}$ denote the edge-vertex incidence matrix of a graph $G$. We say that a flow $f \in \R^E$ routes a demand $d \in \R^V$ if $B^\top f = d$. A flow routing demand $0$ is called a circulation.
For $a, b > 0$ and $\alpha \ge 1$ we write that $a \approx_{\alpha} b$ if $\alpha^{-1}b \le a \le \alpha b$.

\section{Overview}
\label{sec:overview}

In this section we provide an overview of our algorithm. We start by discussing how to reduce SSSP to detecting when vertices' shortest path lengths have decreased by a $(1-\eps)$ factor, and state our main detection algorithm (\cref{thm:detect}). Towards proving \cref{thm:detect}, we review the key points of the incremental mincost flow interior point method (IPM) from \cite{BLS23,CKLMP24}. Finally, we discuss our two new ideas for giving the main detection algorithm.

\subsection{Standard reduction to detecting distance decreases}
\label{sec:reduce}

The following is our main detection theorem. Afterwards, we explain how to apply this theorem to get an approximate incremental SSSP data structure.
\begin{theorem}
\label{thm:detect}
There is an algorithm that takes the following as input: incremental graph $G = \{G^{(t)}\}_{t \in [0,q]}$ with source $s$ such that $d_{G^{(t)}}(s, v) \in [L, 2L]$ for all $t \in [0,q]$ and $v \in V \setminus \{s\}$, parameters $0 < \alpha \le \eps$, and $m = |G^{(q)}|$. The algorithm satisfies the following properties:
\begin{enumerate}
    \item The algorithm is allowed to terminate at an iteration $\tau$ if
    \[ \sum_{v \in V} d_{G^{(\tau)}}(s, v) < (1-\alpha/2)\sum_{v \in V} d_{G^{(0)}}(s, v), \] \label{item:totaldist}
    \item Returns sets $S^{(t)}$ for $t \in [0, \tau)$ such that if $d_{G^{(t)}}(s, v) < (1-\eps) d_{G^{(0)}}(s, v)$ for $t < \tau$, then $v \in S^{(t')}$ for some $t' \le t$, and \label{item:sets}
    \item The sets $S^{(t)}$ are small: $\sum_{t=0}^{\tau-1} |S^{(t)}| \le \frac{\alpha\kappa}{\eps}m$ for some $\kappa \le m^{o(1)}$, independent of $\alpha, \eps$. \label{item:size}
\end{enumerate}
The algorithm is deterministic and runs in total time $m^{1+o(1)}$.
\end{theorem}
Let us explain how to go from this to an SSSP algorithm with almost-linear runtime. This is formally done in \cref{sec:recurse}. First, by standard reductions we may assume that all vertices have indegree and outdegree at most $3$, and that all lengths are polynomially bounded. Also, by modifying the graph by adding a weight $L$ from before the source $s$, and edges of length $2L$ to all vertices, we can assume all shortest path lengths are in $[L, 2L]$. By going over all $L$, this suffices to maintain a $(1+\eps)$-approximation for all original shortest path lengths.

At the start of the algorithm, compute an exact SSSP tree from $s$. Because the graph is incremental, the shortest path lengths can only decrease -- our goal is to detect when this happens and to report that the algorithm is now unsure of the length of the shortest path to these vertices, and all other vertices are guaranteed to have their shortest path lengths at least $(1-\eps)$ of the original value. Let $S$ be the set of vertices that are reported. We can recursively define a graph with all vertices except those in $S$ contracted away (formally, see \cref{def:gds}). This preserves \emph{all} shortest paths to $S$ up to a $(1-\eps)$ factor, because we have a $(1-\eps)$-approximation to the shortest path length of the vertices we have contracted away. Note that the contracted graph is incremental if the original graph is.
Finally, restart whenever \cref{item:totaldist} triggers, and recompute the true SSSP distances, and set $S$ to the empty set.

Thus, it suffices to verify the runtime. \cref{item:totaldist} can trigger at most $O(1/\alpha)$ times. The graph at the next level undergoes at most $\alpha m^{o(1)}$ edges insertions during a phase between recomputations of the SSSP, so the amortized recourse is $m^{o(1)}$. By setting $\alpha$ to be sufficiently small, this becomes an almost-linear time algorithm.

\subsection{Prior work on incremental interior point method}
\label{sec:ipm}

The starting point of our algorithm for \cref{thm:detect} is the recent almost-linear-time algorithm of \cite{CKLMP24} for incremental thresholded mincost flow. Informally, the main result of \cite{CKLMP24} is that given a fixed demand $d$ and an incremental graph $G$, we can find the first edge insertion after which the minimum cost of a flow routing demand $d$ is at most some given threshold $F^*$.

\begin{theorem}[\!\!{\cite[Theorem 1.4]{CKLMP24}}]
\label{thm:incrflow}
There is an algorithm that takes an incremental graph $G$ undergoing $m$ edge insertions with polynomially-bounded capacities and costs, a fixed demand $d$, and threshold $F^*$, and finds the first insertion after which the mincost flow routing demand $d$ has cost at most $F^*$. The algorithm is deterministic and runs in total time $m^{1+o(1)}$.
\end{theorem}

In the case of SSSP, by setting $d = (-(n-1), 1, \dots, 1)$ we can find the first insertion after which $\sum_{v \in V} d_G(s, v) \le F^*$. In our setting, it is most natural to set $F = \sum_{v \in V} d_{G^{(0)}}(s, v)$, i.e., the initial sum of lengths, and set $F^* = (1-\alpha)F$ for some very small $\alpha = m^{-o(1)}$. In other words, we want to find the first time after which the total shortest path length has gone down by $\alpha m$, which by Markov's inequality also implies that at most $O(\alpha m/\eps)$ vertices have their shortest path lengths decrease by $(1-\eps)$ (recall that $d_G(s, v) \in [L, 2L]$ for all $v$).

With this setup, if $\ell$ is the length vector, the approach of \cite{BLS23,CKLMP24} can be summarized as designing an algorithm to decrease the following potential function:
\[ \Phi(f) := 10m \log(\ell^\top f - F^*) + \sum_{e \in E} -\log f_e. \]
This is done by repeatedly approximately solving a \emph{minimum-ratio cycle} problem: $\min_{B^\top \Delta = 0} \frac{g^\top \Delta}{\|W\Delta\|_1}$ where $W = \mathsf{diag}(w)$ for $w_e = 1/f_e$ and $g = \g \Phi(f)$, i.e., $g_E = \frac{10m}{\ell^\top f - F^*} \ell_e - 1/f_e$. The work of \cite{BLS23,CKLMP24} proves that designing a data structure to solve $m^{1+o(1)}$ minimum-ratio cycle problems up to $m^{o(1)}$ approximation over updates to $g, w$ suffices to give an almost-linear time algorithm for incremental mincost flow. Over these $m^{1+o(1)}$ steps, the algorithm maintains approximations $\wt{g} \approx g$ and $\wt{w} \approx w$ that change entries only $m^{1+o(1)}$ total times. We refer to this as \emph{stability} of the IPM. Formally, we maintain a flow $\wt{f}_e$ which changes on edges only $m^{1+o(1)}$ total times such that $w_e|f_e-\wt{f}_e| \le m^{-o(1)}$ for all $e$, and use $\wt{f}_e$ to define $\wt{g}_e$ and $\wt{w}_e$. A deterministic data structure for min-ratio cycles was designed in \cite{CKLMP24} based on $\ell_1$ oblivious routing constructions and fully dynamic APSP. We summarize its guarantees below.

\begin{definition}[\!{\cite[Definition 3.7]{CKLMP24}}]
    \label{def:solver_ds}    
    We call a data structure $\mathcal{D} = \textsc{Solver}(G, w, g, \ell, f, q, \Gamma, \eps)$ initialized with 
\begin{itemize}
    \item a graph $G = (V, E)$ and
    \item weights $w \in \R^{E}_{\geq 0}$, gradients $g \in \R^{E}$, costs $\ell$, a flow $f \in \R^{E}$ routing demand $d$, and 
    \item a quality parameter $q > 0$, a step-size parameter $\Gamma > 0$ and a accuracy parameter $\eps > 0$. 
\end{itemize}
a $\gamma_{\mathrm{approx}}$ min-ratio cycle \textbf{solver} if it (implicitly) maintains a flow vector $f$ such that $f$ routes demand $d$ throughout and supports the following operations. 
\begin{itemize}
    \item $\textsc{ApplyCycle}()$: One of the following happens.  
        \begin{itemize}
            \item Either: The data structure finds a circulation $\Delta \in \R^{E}$ such that $g^\top \Delta / \|W\Delta\|_1 \leq -q$ and $\|W \Delta\|_1 = \Gamma$. In that case it updates $f \gets f + \Delta$ and it returns a set of edges $E' \subseteq E$ alongside the maintained flow values $f_{e'}$ for $e' \in E'$. 

            For every edge $e$, between the times that it is in $E'$ during calls to $\textsc{ApplyCycle}()$, the value of $w_ef_e$ does not change by more than $\epsilon$.
            \item Or: The data structure certifies that there is no circulation $\Delta \in \R^{E}$ such that $g^\top \Delta / \|W\Delta\|_1 \leq -q/\gamma_{\mathrm{approx}}$ for some parameter $\gamma_{\mathrm{approx}}$.
        \end{itemize}
    \item $\textsc{UpdateEdge}(e, w, g)$: Updates the weight and gradient of an edge $e$ that was returned in the set $E'$ by the last call to $\textsc{ApplyCycle}()$.
    \item $\textsc{InsertEdge}(e, w, g, \ell)$: Adds edge $e$ to $G$ with weight $w$, gradient $g$, cost $\ell$. The flow $f_e$ is initialized to $0$.
    \item $\textsc{ReturnCost}()$: Returns the flow cost $\ell^\top f$.
    \item $\textsc{ReturnFlow}()$: Explicitly returns the currently maintained flow $f$. 
\end{itemize}
The sum of the sizes $|E'|$ of the returned sets by $t$ calls to $\textsc{ApplyCycle}()$ is at most $t\Gamma/\epsilon q$. 
\end{definition}

\begin{theorem}[\!\!{\cite[Theorem 3.8]{CKLMP24}}]
\label{thm:mrc}
There is a data structure $\mathcal{D} = \textsc{Solver}(G, w, g, \ell, f, q, \Gamma, \epsilon)$ as in \Cref{def:solver_ds} for $\gamma_{\mathrm{approx}} = e^{O(\log^{167/168} m)}$.
It has the following time complexity for $\gamma_{\mathrm{time}} = e^{O(\log^{167/168} m \log \log m)}$ where $m$ is an upper bound on the total number of edges in $G$:
\begin{enumerate}
    \item The initialization takes time $|E| \gamma_{\mathrm{time}}$.
    \item The operation $\textsc{ApplyCycle}()$ takes amortized time $(|E'| + 1) \cdot \gamma_{\mathrm{time}}$ when $|E'|$ edges are returned.
    \item The operations $\textsc{UpdateEdge}()$/$\textsc{InsertEdge}()/\textsc{ReturnCost}()$ take amortized time $\gamma_{\mathrm{time}}$. 
    \item The operation $\textsc{ReturnFlow}()$ takes amortized time $|E|\cdot \gamma_{\mathrm{time}}$.
\end{enumerate}
\end{theorem}

\subsection{Idea \#1: Using duality of min-ratio cycles}
\label{sec:idea1}

In this section, we start discussing how to prove \cref{thm:detect}. We will informally establish a version of \cref{thm:detect} where everything but \cref{item:size} holds. Instead, it will hold that $\sum_{t=0}^{\tau-1} |S^{(t)}| \le \kappa m,$
so a factor of $1/\alpha$ times larger than what is desired. This increases the resulting recourse by a factor of $\alpha$, which is unacceptable if we want almost-linear runtimes. Later, in \cref{sec:idea2}, we discuss how to build a highly nonstandard IPM to gain back this factor of $\alpha$, by heavily exploiting the structure of the SSSP problem.

Consider the log-barrier potential function as defined earlier, where $F^* = (1-\alpha)F$:
\[ \Phi(f) = 10m \log(\ell^\top f - F^*) + \sum_{e \in E} -\log f_e. \] The min-ratio cycle problem induced by this is $\min_{B^\top \Delta = 0} \frac{g^\top \Delta}{\|W\Delta\|_1}$ where $w = 1/f$ and $g = \frac{10m}{\ell^\top f - F^*} \ell - w$. If there is min-ratio cycle whose quality is at most $q = -1/4$, then \cref{thm:mrc} finds a cycle of quality at most $-\Omega(\gamma_{\mathrm{approx}})$, and thus allows us to decrease $\Phi$ by at least $m^{-o(1)}$. Otherwise, $\min_{B^\top \Delta = 0} \frac{g^\top \Delta}{\|W\Delta\|_1} \ge -1/4$ say. A key idea is that we can immediately conclude that there is a good ``dual'' potential vector $\phi$ for the SSSP problem. To see what this means, recall that a dual potential for shortest paths is a vector $\phi: \R^V \to \R$ such that $\phi(s) = 0$ and $\phi(v) - \phi(u) \le \ell_e$ for all edges $e = (u, v)$. Note that the previous condition implies that $\phi(v) \le d_G(s, v)$ for all vertices $v$. Thus, we can try to use the potential vector $\phi$ as a way to decide if the shortest path to vertex $v$ may have decreased significantly.

Let us describe how to build this dual potential $\phi$. In \cref{lemma:dual}, we use strong duality to prove that if $\min_{B^\top \Delta = 0} \frac{g^\top \Delta}{\|W\Delta\|_1} \ge -1/4$, then there is a vector $\phi$ with $\|W^{-1}(B\phi + g)\|_\infty \le 1/4$. Expanding this out and using the definition of $g$ then gives:
\[ \frac34 w_e \le \frac{10m}{\ell^\top f - F^*}\ell_e - (\phi(v) - \phi(u)) \le \frac54 w_e. \]
Now, intuitively we can set things up so that $\ell^\top f \le (1+\alpha)F$ (i.e., is approximately optimal), so using $w_e = 1/f_e$ gives for $\hphi = \frac{\ell^\top f - F^*}{10m} \phi$ that
\[ 0 \le \ell_e - (\hphi(v) - \hphi(u)) \le \frac{\alpha F}{mf_e}. \]
In other words, $\hphi$ are valid dual potentials. Additionally, one can prove that $\sum_{v \in V} \hphi(v) \ge (1-O(\alpha))F$, i.e., the sum is close to the true shortest path lengths. Now a natural idea would simply be to return ``dangerous" vertices where $\phi(v) \le (1-\eps)d_G(s, v)$ (here, $d_G(s, v)$ refers to the original shortest path length), of which there are at most $O(\alpha m /\eps)$ by Markov. However, there is a critical issue: \emph{we only know such $\phi$ exist, and we do not have access to them.} Furthermore, we do not have bounds on how fast $\phi$ changes: while there are at most $O(\alpha m / \eps)$ bad vertices at each point in time, in principle they can change quickly between iterations.

To get around this, our first idea is to add additional edges to the graph $G$ to learn more from the existence of a dual vector $\phi$. Specifically, we add edges of length $d_G(s, v)$ from $s$ to each vertex $v$ (again, $d_G(s, v)$ refers to the shortest path length in the \emph{original} graph), and give each edge a log-barrier in $\Phi$. Call this set of edges $\hat{E}$. Obviously, this does not affect the shortest path lengths. However, from $\|W^{-1}(B\phi + g)\|_\infty \le 1/4$ we know for all $e = (s, v) \in \hat{E}$:
\[ 0 \le \ell_e - \hphi(v) \le \frac{\alpha F}{mf_e} \le 2\alpha L/f_e, \]
where we have used that $\hphi(s) = 0$ and that $F \le 2Lm$. Thus if $f_e \ge \frac{2\alpha}{\eps}$ then $\ell_e - \hphi(v) \le \eps L$, and thus we know that the shortest path length to $v$ has no decreased by a $(1-\eps)$ factor. Conversely, if $f_e \le \frac{2\alpha}{\eps}$ then we are unsure if the shortest path length to $v$ has decreased. Because all previous bounds are stable up to constant multiplicative factors, an edge going from dangerous to not-dangerous implies that $f_e$ changes by a constant factor. By standard IPM stability analysis discussed earlier, this only can happen $m^{1+o(1)}$ times throughout the algorithm. This explains why the log-barrier algorithm satisfies all properties of \cref{thm:detect} other than \cref{item:size}.

\paragraph{Recourse calculation.} We briefly sketch why the extra factor of $\alpha$ in the recourse is not acceptable towards obtaining an almost-linear runtime. Consider running the detection algorithm in \cref{thm:detect} in phases. The number of phases is about $\O(1/\alpha)$ since each time the algorithm terminates, the total shortest path length decreases by $(1-\alpha)$. In each phase, imagine that the total size of the sets $S^{(t)}$ is $m^{1+o(1)}$ (as opposed to $\alpha m^{1+o(1)}$ as written in \cref{thm:detect}). Then, the total number of updates to the next layer is $O(m^{1+o(1)}/\alpha)$, while the size of the graphs is $O(\alpha m)$. Thus, in $L$ layers the size of $O(\alpha^L m)$ but the total number of updates is $m^{1+o(1)}/\alpha^L$. No matter how we trade this off, it does not obtain an almost-linear runtime. Note that this issue goes away if we save a factor of $\alpha$ in the total number of updates to the next layer, as is claimed in \cref{thm:detect}.

\subsection{Idea \#2: Modified interior point method}
\label{sec:idea2}

Let us try to understand better where the above analysis is lossy. The initial flow assigns $1$ unit on each edge $e \in \hat{E}$ before any steps occur. The logarithmic barrier on edges $e \in \hat{E}$ seems to not penalize movement away from $f_e = 1$ harshly enough. Thus, our idea is to replace the log barrier on $e \in \hat{E}$ with a high-power barrier $x^{-p}$ for $p \approx \log m$. Formally, $\Phi$ is defined in \eqref{eq:pot}, as is chosen to be $x^{-p}/p$ for $x \le 1$, and for $x \ge 1$ interpolates to a logarithmic barrier.
Now, on the each $e$ (with $f_e \le 1$) we set the weight to be $w_e = f_e^{-p-1}$ and $g = \nabla \Phi = \frac{10m}{\ell^\top f - F^*} \ell - w$ still. Note that $w_e$ is negative the gradient of the function $x^{-p}/p$ evaluated at $x = f_e$. This is quite nonstandard from the perspective of IPMs because the function $x^{-p}$ is not self-concordant, so it is very surprising that we are still able to do an IPM algorithm with this potential.

Let us explain some intuition for why such an IPM algorithm should still work in the case of SSSP. On the edge $e \in \hat{E}$ we initially assign $f_e = 1$. As long as the total sum $\sum_{v \in V} d_G(s, v)$ has not decreased by more than a $(1-\alpha)$ factor, assigning $f_e \approx 1$ still maintains a flow which approximately minimizes the sum of distances. Thus, $f_e$ doesn't have to move too far away from $1$, and thus the function $f_e^{-p}$ still looks somewhat stable. 

Formally, we are still able to prove that if the min-ratio cycle problem has a good value, then we can decrease the potential by $\Omega(1/p)$ (see \cref{lemma:progress}) by taking a step $\Delta$ with $\|W\Delta\|_1 \le 1$. The high-power barrier makes it much harder for $f_e^*$ to move away from $1$, and we can prove that the number of edges that ever become dangerous (i.e., shortest path distance went down by a $(1-\eps)$ factor) is now a factor of $\alpha$ smaller than in the case of the log-barrier (see \cref{lemma:hechange} for the formal proof). Informally, here is where the savings of $1/\alpha$ comes from. When we use the barrier $x^{-p}$, the following properties hold:
\begin{itemize}
\item The weight on the edge $e$ is its gradient which is $f_e^{-p-1}$. 
\item An edge is ``dangerous'', i.e., its shortest path length may have decreased, if $f_e^{-p-1} \ge \eps/\alpha$.
\item The ``cost'' to change the flow on edge $e$ by $\Delta_e$ is $O(f_e^{-p-1}|\Delta_e|)$, i.e., its weight times the change.
\end{itemize}
By the third property, changing $f_e$ by an additive $\Delta_e$ changes the weight $f_e^{-p-1}$ by a multiplicative $1+O(p\Delta_e/f_e) = 1+O(pf_e^p \cdot f_e^{-p-1}\Delta_e)$ factor. Thus, the cost of an edge's weight increasing to $\eps/\alpha$ is $\Omega_p(x^{-p})$ where $x^{-p-1} \approx \eps/\alpha$. When $p \approx \log m$, $x \ge 1/2$ so $x^{-p} \approx x^{-p-1} \approx \eps/\alpha$. The total ``cost'' available is at most $m^{1+o(1)}$, which is the number of steps in the algorithm. Thus the number of dangerous edges is bounded by $\alpha m^{1+o(1)}/\eps$, as desired. Note that this bound is in fact tight up to the $m^{o(1)}$, because within a phase where $\sum_{v \in V} d_G(s, v)$ decreased by a factor of $1-\alpha$, if all $d_G(s, v) \in [L, 2L]$, the true number of vertices $v \in V$ whose $d_G(s, v)$ decreased by a $(1-\eps)$ factor is at most $O(\alpha m/\eps)$.

\begin{remark}
This is the first instance the author has seen of using a non-logarithmic barrier within an interior point method that can be successfully analyzed. We are very interested in whether barriers qualitatively different than $-\log x$ can be used in other settings, or whether this is a trick very specific to the SSSP problem.
\end{remark}

\section{Detecting Dangerous Vertices}
\label{sec:detect}
The goal of this section is to establish \cref{thm:detect}. We assume $L = 1$ by scaling. We start by defining the potential function we use to design the IPM. As discussed in \cref{sec:idea1}, we augment the graph $G$ with extra edges to help detect when distances decrease.
\begin{definition}[Augmented edges]
\label{def:augment}
Consider initial distances $d_v = d_{G^{(0)}}(s, v)$. Let $\hat{E}$ be a set of edges $(s, v)$ of length $d_v$ for all $v \in V \setminus \{s\}$.
\end{definition}
We will design a potential function on the graph $G$ with the edges in $\hat{E}$ added. We will always let $E$ be the original edges of the graph $G$.

We now define barrier functions on each edge. Edges $e \in E$ will have the standard log-barrier potential, but edges $e \in \hat{E}$ will have an alternate potential.
\begin{definition}
\label{def:pot}
For a positive integer $p$, define a function $V: \R_{>0} \to \R$ as:
\[ \begin{cases}
    V(x) = x^{-p}/p & \text{ for } x \le 1 \\
    V(x) = 1/p-\log x & \text{ for } x \ge 1.
\end{cases} \]
\end{definition}
Note that $V$ is continuously differentiable. We will take $p = O(\log(1/\alpha))$ in our algorithm. Let $F = \sum_{v \in V} d_v$ and $F^* = (1-\alpha)F$. Define
\begin{equation}
    \Phi(f) = 10m \log(\ell^\top f - F^*) - \sum_{e \in E} \log(f_e+\delta_e) + \sum_{e \in \hat{E}} V(f_e), \label{eq:pot}
\end{equation}
where $\delta_e > 0$ are some ``slack'' parameters to handle a technical issue (think of $\delta_e = 1/m^{10}$). From here, we are ready to define the gradients and lengths of the min-ratio cycle problem.
\begin{definition}
\label{def:mrc}
Given a flow $f$, define the gradient $g(f) = \g \Phi(f)$ and weights $w(f) \in \R^{E \cup \hat{E}}$ as follows: $w(f)_e = \frac{1}{f_e+\delta_e}$ for $e \in E$, and $w(f)_e = -\g V(f_e)$ for $e \in \hat{E}$.
\end{definition}
Note that $g(f) = \frac{10m}{\ell^\top f - F^*} \ell - w(f)$.

\subsection{Single step analysis}
\label{sec:step}

We first prove that there is a step to decrease the potential if $\ell^\top f$ is much larger than $F$. This should intuitively hold, because the total shortest path length is $F$, so there is no real reason the algorithm should be working with flows with total length much greater than $F$.
\begin{lemma}
\label{lemma:value}
Let $f$ be a flow routing demand $d = (n-1, -1, \dots, -1)$ such that $\ell^\top f > (1+5\alpha)F$. Let $g, w$ be defined for $f$ as in \cref{def:mrc}. Then there is a circulation $\Delta$ with $\frac{g^\top \Delta}{\|W\Delta\|_1} \le -1$.
\end{lemma}
\begin{proof}
We will define a flow $f^*$ routing demand $d$ as follows. For an edge $e \in \hat{E}$, if $f_e \ge 1$ then set $f_e^* = 1$. Otherwise, set $f_e^* = f_e$. Now define $f^*$ on $e \in E$ as follows. Let $e_v \in \hat{E}$ correspond to vertex $v$. Then add to $f^*$ a shortest path flow consisting of $1 - f_{e_v}^*$ units from $s \to v$. Define $\Delta = f^* - f$ and note that
\[ \ell^\top \Delta = \ell^\top f^* - \ell^\top f \le \sum_{v \in V} d_G(s, v) - \ell^\top f \le F - \ell^\top f. \]
Thus the gradient term can be bounded as
\begin{align*}
    g^\top \Delta &= \frac{10m}{\ell^\top f - F^*} \ell^\top \Delta - \sum_{e \in E} \frac{\Delta_e}{f_e+\delta_e} - \sum_{e \in \hat{E}} w_e\Delta_e \\
    &\le -8m + \sum_{e \in E} 2 - \frac{|\Delta_e|}{f_e+\delta_e} + \sum_{e \in E} 2 - w_e|\Delta_e| \\
    &\le -4m - \|W\Delta\|_1,
\end{align*}
as desired. In the first inequality, we have used that
\[ \frac{\ell^\top \Delta}{\ell^\top f - F^*} \le \frac{\ell^\top f - F}{\ell^\top f - F^*} \le -4/5, \]
because $F^* = (1-\alpha)F$ and $\ell^\top f \ge (1+5\alpha)F$.
Additionally, we have used the inequality $-w_e\Delta_e \le 2 - w_e|\Delta_e|$ for $e \in \hat{E}$. To show this, we consider cases based on if $f_e \le 1$. If $f_e \le 1$ then $\Delta_e = 0$ so the claim is evident. If $f_e \ge 1$ then $f_e^* = 1$ and $w_e = 1/f_e$, so
\[ -w_e\Delta_e = \frac{f_e-1}{f_e} \le 2 - |1-f_e|/f_e, \] because $f_e \ge 1$. This completes the proof.
\end{proof}

Next, we prove that solving the min-ratio cycle problem $\min_{B^\top \Delta = 0} \frac{\wt{g}^\top \Delta}{\|\wt{W}\Delta\|_1}$ for approximate gradients $\wt{g}$ and $\wt{W}$ is sufficient to decrease the potential.
\begin{lemma}
\label{lemma:progress}
Let $\beta < 1$, and $f$ be a flow with gradient $g$ and weights $w$. Let $\wt{g}$ satisfy that $\|W^{-1}(g - \wt{g})\|_\infty \le \beta/10$, $\wt{w} \approx_2 w$, and circulation $\Delta$ satisfy that $\frac{\wt{g}^\top \Delta}{\|\wt{W}\Delta\|_1} \le -\beta$. Let $\eta$ be such that $\eta \|\wt{W}\Delta\|_1 = \frac{\beta}{100p}$. Then
\[ \Phi(f + \eta\Delta) \le \Phi(f) - \Omega(\beta^2/p). \]
\end{lemma}
\begin{proof}
We first establish that $\frac{g^\top\Delta}{\|W\Delta\|_1} \le -\beta/3$. We may assume by scaling that $\|\wt{W}\Delta\|_1 = 1$. Then
\[ g^\top \Delta = \wt{g}^\top \Delta + (g-\wt{g})^\top \Delta \le -\beta + \|\wt{W}^{-1}(g-\wt{g})\|_\infty \|\wt{W}\Delta\|_1 \le -\beta + \beta/5 \le -4\beta/5. \]
The bound follows by combining this with $\|W\Delta\|_1 \le 2$.

We consider how adding $\eta\Delta$ to $f$ affects each term of $\Phi$. By concavity of $\log$, we get that
\[ 10m \log(\ell^\top(f+\eta\Delta) - F^*) \le 10m \log(\ell^\top f - F^*) + \eta \frac{10m}{\ell^\top f - F^*} \ell^\top \Delta. \]
Next consider $e \in E$, where $w_e = \frac{1}{f_e+\delta_e}$. Because $w_e|\eta\Delta_e| \le 1/100$ for all $e$, we get that
\begin{align*}
    -\log(f_e+\delta_e+\eta\Delta_e) \le -\log(f_e+\delta_e) - \eta w_e\Delta_e + \eta^2w_e^2\Delta_e^2.
\end{align*}
For $e \in \hat{E}$, we will check that
\begin{equation}
V(f_e+\eta\Delta_e) \le V(f_e) + \eta \g V(f_e) \Delta_e + 2p \eta^2 w_e^2 \Delta_e^2. \label{eq:want2nd}
\end{equation}
To prove this, we split into cases based on $f_e$. We first consider $f_e \ge 1$. Then $w_e = 1/f_e$ and thus $\eta|\Delta_e|/f_e \le \frac{1}{100p}$. Note that for all $x \in [f_e - \eta\Delta_e, f_e + \eta\Delta_e]$ that $\g^2 V(x) \le 3p/f_e$. Indeed, this is clear if $f_e \ge 1.1$, and if $f_e \le 1.1$ then for all $1 > x \ge (1-1/(10p))$ we know that $\g^2 V(x) = (p+1) x^{-p-2} \le 2p$. Thus \eqref{eq:want2nd} follows by Taylor's theorem.

For the case $f_e \le 1$, we prove the stronger bound
\[ V(f_e+\eta\Delta_e) \le V(f_e) + \eta \g V(f_e) \Delta_e + 2p \eta^2 f_e^{-(p+2)} \Delta_e^2. \]
This is stronger because $f_e^{-(p+2)} \le f_e^{-2p-2} = w_e^2$. Again, it suffices to check that $\g^2 V(x) \le 4p f_e^{-(p+2)}$ for $x \in [f_e - \eta \Delta_e, f_e + \eta \Delta_e]$, where we know that
$\eta\Delta_e \le \frac{1}{100p} f_e^p \le \frac{f_e}{100p}$ because $w_e|\Delta_e| \le \frac{1}{100p}$.
This is done similarly to above.

Summing these bounds gives us
\begin{align*}
    \Phi(f + \eta\Delta) - \Phi(f) &\le \eta g^\top \Delta + 2p\eta^2 \sum_{e \in E \cup \hat{E}} w_e^2 \Delta_e^2 \le \eta(g^\top \Delta + 2p\eta \|W\Delta\|_1^2) \\
    &\le \eta\left(\frac{-\beta}{3} \cdot \|W\Delta\|_1 + \frac{\beta}{25} \|W\Delta\|_1 \right) \le -\Omega(\beta^2/p),
\end{align*}
where we have used that $2p\eta \|W\Delta\|_1 \le 4p\eta \|\wt{W}\Delta\|_1 \le \beta/25$.
\end{proof}
Finally we prove that if the true min-ratio cycle problem has a good value, then the approximate version also has a good value.
\begin{lemma}
\label{lemma:approx}
Let $g, w$ be vectors such that there is a circulation $\Delta$ with $\frac{g^\top \Delta}{\|W\Delta\|_1} \le -\beta$. Let $\wt{w}$ satisfy that $\wt{w} \approx_2 w$, and $\wt{g}$ satisfy that $\|W^{-1}(\wt{g} - g)\|_\infty \le \beta/3$. Then $\frac{\wt{g}^\top \Delta}{\|\wt{W}\Delta\|_1} \le -\beta/3$.
\end{lemma}
\begin{proof}
By normalizing, we may assume that $\|W\Delta\|_1 = 1$, and thus $\|\wt{W}\Delta_1\|_1 \le 2$. Then the conclusion follows from
\[ \wt{g}^\top \Delta \le g^\top \Delta + (\wt{g} - g)^\top \Delta \le -\beta + \|W^{-1}(\wt{g}-g)\|_\infty \|W\Delta\|_1 \le -2\beta/3. \qedhere \]
\end{proof}

\subsection{Duality}
\label{sec:duality}

When there is no min-ratio cycle to make progress, we get that there is a potential vector $\phi$ such that the gradient $g$ is close to $B\phi$ in the proper norm. Later, we will use this to read off edges $e \in \hat{E}$ where the flow value changes significantly, and thus may signal a vertex whose length has gone down by a $(1-\eps)$.
\begin{lemma}
\label{lemma:dual}
Let $G$ be a graph and consider vectors $g \in \R^E$, $w \in \R^E_{>0}$, and $W = \diag(w)$. If there is no circulation $\Delta \in \R^E$ with $\frac{|g^\top \Delta|}{\|W\Delta\|_1} \le -\beta$, then there is $\phi \in \R^V$ such that $\|W^{-1}(B\phi + g)\|_\infty \le \beta$.
\end{lemma}
\begin{proof}
By strong duality, we know that
\[
1/\beta \le \min_{\substack{B^\top \Delta = 0 \\ g^\top \Delta = -1}} \|W\Delta\|_1 = \min_{\substack{B^\top \Delta = 0 \\ g^\top \Delta = -1}} \max_{\|z\|_\infty \le 1} z^\top W\Delta = \max_{\|z\|_\infty \le 1} \min_{\substack{B^\top \Delta = 0 \\ g^\top \Delta = -1}} z^\top W\Delta. \]
For this to not be undefined, we need for $Wz = B\phi + \lambda g$ for some $\lambda$: thus we get
\[ 1/\beta \le \max_{\|W^{-1}(B\phi + \lambda g)\|_\infty \le 1} \lambda. \]
Scaling by $\lambda$ completes the proof.
\end{proof}

We will now prove that if $\phi$ is a dual vector from \cref{lemma:dual} where $\phi(s) = 0$, then scaling $\phi$ produces a vector $\hphi$ satisfying (1) $\hphi$ is a valid dual vector in that $\hphi(v) - \hphi(u) \le \ell_e$ for any $e = (u, v)$, and thus $\hphi(v) \le d_G(s, v)$ for all $v$, and (2) $\sum_{v \in V} \hphi(v) \ge (1-O(\alpha))F$, i.e., the $\hphi$ are large lower bounds on average.

\begin{lemma}
\label{lemma:phiprop}
Let flow $f$ induce gradient $g$ and weights $w$ (as in \cref{def:mrc}), and assume that
\[ (1-\alpha/2)F \le \ell^\top f \le (1+5\alpha)F, \] and $\Phi(f) \le 20m \log m$. Let $\phi \in \R^V$ satisfy that $\phi(s) = 0$ and $\|W^{-1}(B\phi + g)\|_{\infty} \le \beta$ for $\beta \le 1/2$. Then for $\hphi := \frac{\ell^\top f - F^*}{10m}\phi$ it holds that
\begin{equation}
    \frac12 \frac{\ell^\top f - F^*}{10m}w_e \le \ell_e - (\hphi(v) - \hphi(u)) \le 2\frac{\ell^\top f - F^*}{10m}w_e \enspace \text{ for } \enspace e = (u, v) \in E \cup \hat{E}, \label{eq:phivalid}
\end{equation}
and $\hphi(v) \le d_G(s, v)$ for all $v \in V$, and $\sum_{v \in V} \hphi(v) \ge (1-100\alpha p \log m)F$.
\end{lemma}
\begin{proof}
We start by establishing \eqref{eq:phivalid}. Indeed, for an edge $e = (u, v)$, $\|W^{-1}(B\phi + g)\|_\infty \le \beta$ tells us
\[ -\beta w_e \le \frac{10m}{\ell^\top f - F^*} \ell_e - w_e - (\phi(v) - \phi(u)) \le \beta w_e. \]
For $\beta \le 1/2$, rearranging yields \eqref{eq:phivalid}. In particular, $\hphi(v) - \hphi(u) \le \ell_e$ for all edges $e$, and thus if $\hphi(s) = 0$ then $\hphi(v) \le d_G(s, v)$ for all $v \in V$.
For the final point, note that
\begin{align*}
    \sum_{v \in V} \hphi(v) &= d^\top \hphi = (B\hphi)^\top f \ge \ell^\top f + (B\hphi - \ell)^\top f \\
    &\ge (1-\alpha/2)F - \sum_{e \in E \cup \hat{E}} 2\frac{\ell^\top f - F^*}{10m}w_ef_e \\
    &\ge (1-\alpha/2)F - \frac{6\alpha F}{10m} \left(\sum_{e \in E} \frac{f_e}{f_e+\delta_e} + \sum_{e \in \hat{E}} 1 + \frac{1}{f_e^p} \right) \\
    &\ge (1-\alpha/2)F - \frac{6\alpha F}{10m}\left(m + \sum_{e \in \hat{E}} \frac{1}{f_e^p} \right).
\end{align*}
The result follows because
\[ \frac{1}{p}\sum_{e \in \hat{E}} \frac{1}{f_e^p} = \Phi(f) - 10m \log(\ell^\top f - F^*) + \sum_{e \in E} \log(f_e + \delta_e) \le 100m \log m \]
because $f_e \le m^{10}$ for all $e \in E$, and the hypothesis $\Phi(f) \le 20m\log m$.
\end{proof}

\begin{remark}
Our algorithm does not actually require the final claim that \[ \sum_{v \in V} \hphi(v) \ge (1-100\alpha p \log m)F, \]
i.e., that $\hphi$ are potentials which are not much smaller than the true distances on average. However, we believe that including a proof of this statement is important for intuition and may be useful for future works.
\end{remark}

\subsection{Stability}
\label{sec:stability}

In this section we show various forms of stability. This includes the standard properties that the gradient and weights are slowly changing (they only need to be updated very infrequently). A key insight in our analysis is that for the edges $e \in \hat{E}$, the $p$-th power barrier $x^{-p}$ makes it extremely costly to move $x$ away from $1$ at all, and this allows us to improve the recourse (i.e., sum of size of $|S^{(t)}|$) by a factor of $\alpha$ in \cref{thm:detect}.

\begin{lemma}
\label{lemma:wtgw}
Let $f$ be a flow with gradient $g$ and weights $w$ (as in \cref{def:mrc}). Let $r \approx_{1+\beta/100} \ell^\top f - F^*$. If $|\wt{f}_e - f_e| \le \frac{\beta}{100p} (f_e+\delta_e)$ for all edges $e$ ($\delta_e = 0$ for $e \in \hat{E}$), then $w(\wt{f}) \approx_{1+\beta/2} w(f)$. Additionally, for $\wt{g} = \frac{10m}{r}\ell - w(\wt{f})$
it holds that
\[ \left\|W^{-1}\Big(\wt{g} - \frac{\ell^\top f - F^*}{r}g\Big)\right\|_\infty \le \beta. \]
\end{lemma}
\begin{proof}
The claim $w(\wt{f}) \approx_{1+\beta/2} w(f)$ follows from the definition of $w$ (see \cref{def:mrc}). Recall that $g = \frac{10m}{\ell^\top f - F^*}\ell - w$. Using this gives
\begin{align*}
    \left|\wt{g}_e - \frac{\ell^\top f - F^*}{r_e}g_e\right| &= \left|w(\wt{f})_e - \frac{\ell^\top f - F^*}{r_e} w_e\right| \le \beta|w_e|,
\end{align*}
where we have used that $r \approx_{1+\beta/100} \ell^\top f - F^*$ and that $w(\wt{f}) \approx_{1+\beta/2} w(f)$.
\end{proof}

Now we come to a critical bound in our analysis: the number of edges $e \in \hat{E}$ whose weight can go above a given threshold $K$. This is critical for our analysis, and uses the structure of the barrier on edges $e \in \hat{E}$.

\begin{lemma}
\label{lemma:hechange}
Let $f^{(0)}, \dots, f^{(T)}$ be a sequence of flows with weights $w^{(t)} = w(f^{(t)})$ defined as in \cref{def:mrc}, $W^{(t)} = \diag(w^{(t)})$, and for $0 \le t \le T$ define $\Delta^{(t)} = f^{(t+1)} - f^{(t)}$ where $\|(W^{(t)}) \Delta^{(t)}\|_1 \le \frac{1}{100p}$. Let $f^{(0)}_e = 1$ for $e \in \hat{E}$. For $K > 4$, let $F \subseteq \hat{E}$ consist of edges $e$ such that $w^{(t)}_e \ge K$ for some $0 \le t \le T$. Then $|F| \le O(K^{-\frac{p}{p+1}} T)$.
\end{lemma}
\begin{proof}
Note that $w^{(0)}_e = 1$. For $e \in F$ let $t$ be the first time step where $w^{(t)}_e \ge K$ and let $t' < t$ be the last time step before time $t$ that $f_e^{(t')} \le 1$. Then for $t' < s \le t$, we know that$w^{(s)}_e = \frac{1}{(f^{(s)}_e)^{p+1}}$, and thus we get that
\[ w^{(s+1)}_e \le w^{(s)}_e + 2p \frac{|\Delta^{(s)}_e|}{f^{(s)}_e} w^{(s)}_e \le w^{(s)}_e + 2pK^{\frac{1}{p+1}} w_e^{(s)}|\Delta_e^{(s)}|, \]
where we have used that $f_e^{(s)} = (w_e^{(s)})^{-\frac{1}{p+1}} \ge K^{-\frac{1}{p+1}}$.
Thus if edge $e \in F$, its contribution to the sum $\sum_s \|W^{(s)} \Delta^{(s)}\|_1$ is at least
\[ \frac{K-1}{2p K^{\frac{1}{p+1}}} \ge \Omega(K^{p/(p+1)}/p). \]
We know that $\sum_{s=0}^T \|W^{(s)}\Delta^{(s)}\|_1 \le O(T/p)$, so we conclude that $|F| \le O(K^{-\frac{p}{p+1}}T)$ as desired.
\end{proof}

\subsection{Detection algorithm}
\label{sec:detectalgo}

In this section we give an algorithm to prove \cref{thm:detect}. Towards this, we initialize the data structure in \cref{thm:incrflow} to detect the first insertion $\tau$ during which
\[ \sum_{v \in V} d_{G^{(\tau)}}(s, v) < (1-\alpha/2)F \enspace \text{ for } \enspace F = \sum_{v \in V} d_{G^{(0)}}(s, v). \]
Formally, this is done by setting the demand to have $-(n-1)$ on the source $s$ and $+1$ on each other vertex. The cost of each edge is its length.
The algorithm terminates when this happens. This contributes runtime $m^{1+o(1)}$, and we may assume that $\sum_{v \in V} d_{G^{(\tau)}}(s, v) \ge (1-\alpha/2)F$ at all times.

\paragraph{Algorithm description.} We state the detection algorithm now. Let $d$ be the demand with $-(n-1)$ on $s$ and $+1$ on each other vertex. When a new edge is inserted its cost is equal to its length, its backwards capacity is $\delta = m^{-O(1)}$ and forwards capacity is infinite. Initialize $f^{(0)}$ to have $f^{(0)}_e = 1$ for $e \in \hat{E}$ and $0$ otherwise, so that $f^{(0)}$ routes the demand $d$. It holds that $\Phi(f^{(0)}) \le 20m\log m$. Let $\gamma = \gamma_{\mathrm{approx}} = m^{-o(1)}$ be the approximation ratio of \cref{thm:mrc}.

\paragraph{Decreasing $\Phi$.} Run \cref{thm:mrc} on a graph with approximate gradients and weights $\wt{g}$ and $\wt{w}$, defined as follows. Maintain $r \approx \ell^\top f - F^*$. Let $\wt{f} \in \R^E$ be a vector satisfying that $w_e|\wt{f}_e - f_e| \le \frac{\gamma}{1000p}$ for all $e \in E$, define $\wt{w} = w(\wt{f})$ and $\wt{g} = \frac{10m}{r}\ell(\wt{f}) - \wt{w}$ as in \cref{lemma:wtgw}. Use the operations $\textsc{UpdateEdge}$, $\textsc{InsertEdge}$, to update the gradients and weights on edges to $\wt{g}$ and $\wt{w}$. Use $\textsc{ApplyCycle}$ to solve the resulting min-ratio cycle problems for $\Gamma := \frac{\gamma}{100p}$, i.e., push circulation $\eta \Delta$ where $\eta\|\wt{W}\Delta\|_1 = \frac{\gamma}{100p}$ (see \cref{lemma:progress}). If the value of the returned min-ratio cycle is greater than $q := -\gamma/10$, do not push flow and end this step.

\paragraph{Sending edges to $S^{(t)}$.} After ending an iteration, take any edges with $\wt{w}_e \ge \frac{\eps}{10\alpha}$ that are not part of a previous $S^{(t)}$ and add them to $S^{(t)}$. Return $S^{(t)}$.

\paragraph{Maintaining polynomially bounded lengths.} If $f_e+\delta_e \le \delta$ (recall $\delta = m^{-O(1)}$ is some fixed parameter), then also update $\delta_e \leftarrow \delta_e + \delta$. Then update the edge's gradient and weight. Note that this only decreases $\Phi$ and ensures that $w(f)_e \le m^{O(1)}$ for all $e$ always.

Now we prove that this algorithm satisfies all the desired properties of \cref{thm:detect}.
\begin{proof}[Proof of \cref{thm:detect}]
Let us first bound the number of steps of decreasing $\Phi$. Initially, $\Phi(f) \le 20m\log m$ and the algorithm will terminate when $\Phi(m) \le -100m \log m$. By \cref{lemma:progress}, each step decreases $\Phi$ be $\Omega(\gamma^2/p)$, so the number of steps is at most $O(mp\gamma^{-2}\log m)$.

For $p = O(\log m)$, \cref{lemma:hechange} (and $\wt{w} \approx_2 w$) tells us that the total size of the $S^{(t)}$ is at most
\[ \left(\frac{20\alpha}{\eps} \right)^{\frac{p}{p+1}} \cdot O(mp\gamma^{-2}\log m) \le O(m\gamma^{-2}(\log m)^2 \alpha/\eps), \] which is at most $\alpha m^{1+o(1)}/\eps$ as claimed by \cref{thm:detect} \cref{item:size}.

Finally, we prove \cref{item:sets} of \cref{thm:detect}. Indeed, because the quality of the min-ratio cycle is at least $-\gamma/1000$, we know that $\min_{B^\top \Delta = 0} \frac{g^\top \Delta}{\|W\Delta\|_1} \ge -1/3$ by the approximation ratio of \cref{thm:mrc} and \cref{lemma:approx}. Combining \cref{lemma:dual,lemma:phiprop} tells us that there is a vector $\hphi$ such that $\hphi(v) \le d_G(s, v)$ and $\ell_e - \hphi(v) \in \frac{\ell^\top f - F^*}{10m} \cdot [w_e/2, 2w_e]$ for all $e \in \hat{E}$.
By \cref{lemma:value} we know that $\ell^\top f - F^* \le (1+5\alpha)F - (1-\alpha)F \le 6\alpha F$, and $\ell^\top f - F^* \ge \alpha F/2$ because the algorithm has not terminated. Thus if $d_G(s, v) \le (1-\eps)\ell_e$ then
\[ \eps L \le \ell_e - \hphi(v) \le \frac{\ell^\top f - F^*}{10m} \cdot 2w_e \le \frac{6\alpha mL}{10m} \cdot 2w_e \le 2\alpha Lw_e. \]
Thus, $w_e \ge \frac{\eps}{2\alpha}$ and $\wt{w}_e \ge \frac{\eps}{4\alpha}$, and thus all such edges are added to $S^{(t)}$ by the algorithm description.

To bound the runtime, we first bound how many updates are made to $\wt{g}$ and $\wt{w}$ throughout the execution of the algorithm which satisfy the hypotheses of \cref{lemma:progress,lemma:approx}. We know that maintaining $\|W^{-1}(\wt{f} - f)\|_\infty \le \frac{\gamma}{1000p}$ for $\gamma = \exp(-O(\log^{167/168} m))$ suffices by \cref{lemma:wtgw}. This is ensured by running $\textsc{ApplyCycle}()$ of \cref{thm:mrc} using $\eps = \frac{\gamma}{1000p}$. The runtime of ensuring this is at most $\frac{t\Gamma}{\eps q} \le m^{1+o(1)}$ because $t \le m^{1+o(1)}$ as above, $\Gamma \le m^{o(1)}$, and $\eps \ge m^{-o(1)}$ by the parameter choices.
Finally, we use \cref{thm:incrflow} to detect if \cref{item:totaldist} has occurred, for additional $m^{1+o(1)}$ time.
\end{proof}

\section{Recursive Algorithm}
\label{sec:recurse}
In this section we use \cref{thm:detect} to give our incremental SSSP data structure in \cref{thm:main}. This is mostly standard. We assume that the maximum length $W$ is polynomially bounded, as there is a standard way to go from that case to the general one while paying an extra $\log W$ factor\footnote{For instance, one can consider a length scale $[L, 2L]$, and increase all small edge lengths to at least $L/m^2$, and delete all edges of length greater than $2L$, and run the incremental data structure for each $L = 2^i$ for $i \le O(\log mW)$.}.

Starting with a graph $G$, we want to transform it into a graph on which we can apply \cref{thm:detect}. Thus, we need to make all shortest path lengths be in the range $[L, 2L]$.
\begin{definition}
\label{def:=l}
Given a directed graph $G$ with $n$ vertices and source $s$, define the graph $G_{\aL}$ as follows. Add a new vertex $s'$ with the following outedges: (1) an edge $(s', s)$ of length $L$, and (2) edges $(s', v)$ of length $2L$ for all $v \in V$.
\end{definition}

Also, we need to define a \emph{contracted graph} $G(d, S)$. This takes a graph $G$ with distance estimates $d_v$ for $v \in V$, and restricts the vertex set down to $S \subseteq V$.
\begin{definition}
\label{def:gds}
Given a graph $G$ with subset $S \subseteq V$ source $s \in S$, and distance estimates $d_v$ for $v \in V$, define the graph $G$ with vertex set $S$ as follows. For $v \in S$, add an edge $(s, v)$ of length $d_v$. For all inedges $e = (u, v)$, add an edge $(s, v)$ of length $d_u + \ell_e$. If additionally $u \in S$, add edge $(u, v)$ of length $\ell_e$ to $G$.
\end{definition}

Now we are ready to give the algorithm. We assume in the input of \cref{algo:incrSSSP} that all vertices have indegree and outdegree at most $3$, and all lengths are integer and polynomially bounded (this is without loss of generality). 

\begin{algorithm}[!ht]
\caption{Recursive incremental SSSP}
\label{algo:incrSSSP}
\SetKwProg{Globals}{global variables}{}{}
\SetKwProg{Proc}{procedure}{}{}
\Globals{}{
    $\cD_{\aL}$ for $L = 2^k$ for $0 \le k \le O(\log m)$ -- detection data structures from \cref{thm:detect}. \\
    $\alpha, \kappa, \bar{\eps} := \frac{\eps}{10\log m}$ -- parameters in the algorithm. \\
    $G(d, S)$: graph formed from $G$ with distance estimates $d$ and subset $S \subseteq V$ (see \cref{def:gds}). \\
    $\cD^{(\sssp)}$ -- recursive SSSP data structure on $G(d, S)$.
}
\Proc{$\textsc{Initialize}(G, \eps)$}{
    Compute $d_G(s, v)$ for all $v \in V$, update values. \label{line:compd} \\
    Clear $\cD^{(\sssp)}$, and set $S \assign \emptyset$.
    For each $L = 2^k$, initialize $\cD_{\aL}$ on the graph $G_{\aL}$ (see \cref{def:=l}). \\
    For each $L = 2^k$, pass the vertices $\cD_{\aL}.S^{(0)}$ to $\cD^{(\sssp)}$, apply $S \assign S \cup \cD_{\aL}.S^{(0)}$. \label{line:pass1} 
}
\Proc{$\textsc{Insert}(e)$}{
    For each $L = 2^k$, pass the edge insertion to $\cD_{\aL}$ on the graph $G_{\aL}$. \\
    Pass the edge insertion to $\cD^{(\sssp)}$. \\
    If any $\cD_{\aL}$ terminates, then call $\textsc{Initialize}(G, \eps)$. \\
    For each $L = 2^k$, pass the vertices $\cD_{\aL}.S^{(t)}$ to $\cD^{(\sssp)}$, where $t$ is the number of iterations since the last call to $\textsc{Initialize}$, apply $S \assign S \cup \cD_{\aL}.S^{(t)}$. \label{line:pass2} \\
    Use $\cD^{(\sssp)}$ to update some values of $d_G(s, v)$.
}
\end{algorithm}

\paragraph{Algorithm description.} We initialize detection data structures (\cref{thm:detect}) on each of the graphs $G_{\aL}$. We rebuild every data structure if the sum of shortest path lengths in any of these graphs falls by a $(1-\alpha/2)$ factor (this is when the data structure \emph{terminates}). At this time, we also recompute all shortest paths in the graph $G$ and update the values.

Between rebuilds, each of the data structures $G_{\aL}$ is detecting sets of vertices whose shortest path lengths may have decreased by a $(1-\beps)$ factor, and those are passed to the set $S$ which is used to update the contracted graph $G(d, S)$ (\cref{def:gds}). We recursively invoke the SSSP data structure of \cref{algo:incrSSSP} on the graph $G(d, S)$ to update the shortest path lengths of the vertices in $S$. We can prove that the shortest paths to $S$ in $G(d, S)$ are at most a factor of $(1+\beps)$ larger.

\paragraph{Analysis.} We first define a \emph{phase} in the execution of \cref{algo:incrSSSP}. A phase is the iterations between two consecutive calls to \textsc{Initialize}.
We first observe that the graphs $G(d, S)$ and $G_{\aL}$ are incremental, and that all shortest path distances in $G_{\aL}$ are in $[L, 2L]$.

\begin{lemma}
\label{lemma:gincr}
The graphs $G_{\aL}$ are incremental and satisfy that $d_{G_{\aL}}(s, v) \in [L, 2L]$ for all $v \neq s$. During a phase, the graph $G(d, S)$ is incremental.
\end{lemma}
\begin{proof}
The claims about $G_{\aL}$ are clear by \cref{def:=l}. Within a phase, the graph $G(d, S)$ only undergoes edge insertions when $G$ does or when $S$ gets a new vertex -- both are incremental.
\end{proof}

We next observe that if a vertex has its distance go down by a $(1-O(\bar{\eps}))$ factor during a phase, then it will be inserted into the set $S$.
\begin{lemma}
\label{lemma:svalid}
Consider a phase between calls to \textsc{Initialize} in \cref{algo:incrSSSP}. Let $G^{(t)}$ be the current graph at step $t$, and let $G$ be the graph at initialization. For vertices $v \notin S$, it holds that $d_{G^{(t)}}(s, v) \ge (1-4\beps) d_G(s, v)$.
\end{lemma}
\begin{proof}
Let $L$ be such that $d_G(s, v) \in [L/2, L]$. Then $d_{G_{\aL}}(s, v) = L + d_G(s, v)$ and $d_{G^{(t)}_{\aL}}(s, v) = L + d_{G^{(t)}}(s, v)$. If $v \notin S$, then we know that $d_{G^{(t)}_{\aL}}(s, v) \ge (1-\beps) d_{G_{\aL}}(s, v)$ by \cref{item:sets} of \cref{thm:detect}, and thus
\begin{align*}
    d_{G_{\aL}}(s, v) = d_{G^{(t)}_{\aL}}(s, v) - L \ge d_G(s, v) - 2\beps L \ge (1-4\beps) d_G(s, v),
\end{align*}
as desired.
\end{proof}

To finish up the correctness analysis, it suffices to argue that for $v \in S$ that the shortest path length in $G(d, S)$ is at most a $(1+O(\bar{\eps}))$-factor larger than the true value.

\begin{lemma}
\label{lemma:gdsvalid}
Let $G' = G(d, S)$. For $v \in S$ it holds that $d_{G'}(s, v) \le (1+5\beps)d_G(s, v)$.
\end{lemma}
\begin{proof}
Consider a shortest paths from $s \to v$ and let $u$ be the last vertex on the path that is not in $S$. All edges in the path $u \to v$ have the same length in $G'$ and $G$. Additionally, $G'$ has an edge from $s \to u$ of length $d_{G^{(0)}}(s, v)$ (where $G^{(0)}$ is the graph at initialization time), and by \cref{lemma:svalid} we know that $d_{G^{(0)}}(s, v) \le \frac{1}{1-4\beps} d_G(s, v) \le (1+5\beps) d_G(s, v)$. Thus $d_{G'}(s, v) \le (1+5\beps) d_G(s, v)$.
\end{proof}

Now we have all the pieces to prove our main result \cref{thm:main}.
\begin{proof}[Proof of \cref{thm:main}]
We first analyze the correctness of \cref{algo:incrSSSP}, and then the total runtime. Let $B$ be the number of recursive layers, so that $\alpha = m^{1/B}$.

\paragraph{Correctness.} Consider the top recursive layer. For vertices $v \notin S$, the distance $d_G(s, v)$ computed at initialization in line \ref{line:compd} is at most a $(1+5\beps)$-factor larger than the true value. By \cref{lemma:gincr}, applying $\cD^{(\sssp)}$ recursively on $G(d, S)$ is valid (because it is an incremental graph), and distances in $\cD^{(\sssp)}$ are at most a factor of $(1+5\beps)$ larger by \cref{lemma:gdsvalid}. Thus, the overall approximation ratio is $(1+5\beps)^B \le (1+\eps)$.

\paragraph{Reporting paths.} We briefly discuss how to augment \cref{algo:incrSSSP} to report paths as required by \cref{thm:main}. The paths are reported recursively, where the shortest path in the bottom level $B$ graph is recomputed each time. Consider a level $i \le B$ and and a shortest path from $s$ to $v$ in this graph. If the first edge of this shortest path is a contracted edge, then recursively report the shortest path corresponding to that contracted edge in the previous level. Clearly, this path reporting scheme satisfies the requirements of the theorem.

\paragraph{Runtime.} We prove that the number of phases at the top level is at most $O(\frac{\log m}{\alpha})$, thus at the $b$-th level (top level is $b = 0$), the $\cD^{(\sssp)}$ data structure is initialized at most $O(\frac{\log m}{\alpha})^b$ times in total. Indeed, the bound on the number of phases follows because by \cref{item:totaldist} of \cref{thm:detect}, a new phase starts when some $G_{\aL}$ has $\sum_{v \in V} d_{G_{\aL}}(s, v)$ decrease by a $(1-\Omega(\alpha))$ factor -- this can only happen $O(1/\alpha)$ times. Thus, across all values of $L$, this can happen a total of $O(\frac{\log m}{\alpha})$ times.

By \cref{item:sets} of \cref{thm:detect} and that degrees are $O(1)$ over all $L = 2^k$, during each phase the data structure $\cD^{(\sssp)}$ undergoes $O\left(\frac{\alpha\kappa\log m}{\beps} \right) \cdot m$ edge insertions. Recursively, each instance of $\cD^{(\sssp)}$ at level $b$ undergoes $O\left(\frac{\alpha\kappa\log m}{\beps} \right)^b \cdot m$ edge insertions. By the bound on the number of phases, and the runtime bound of \cref{thm:detect}, the total runtime of level $b$ is at most
\[ m^{o(1)} \cdot O\left(\frac{\alpha\kappa\log m}{\beps} \right)^b \cdot m \cdot O\left(\frac{\log m}{\alpha}\right)^{b+1} \le m^{1+1/B+o(1)} \cdot O\left(\frac{\kappa \log^2 m}{\eps} \right)^B, \]
where we have used that $\alpha = m^{1/B}$. Because $\max\{\kappa, \eps^{-1}\} \le m^{o(1)}$ by assumption, there is a choice of $B$ that makes the total runtime at most $m^{1+o(1)}$, as desired.
\end{proof}

\paragraph{Remarks.} It would be interesting to see if the ideas in this paper can be used to give other incremental/decremental algorithms that are slightly beyond what was achieved in \cite{CKLMP24,BCK+24}: incremental topological sort, decremental SSSP, etc. Additionally, it would be interesting to better understand what the role of different barriers such as $x^{-p}$ are in the context of interior point methods for graph problems and beyond.

Finally, for $\eps \ge \exp(-O(\log m)^{167/168})$ the runtime of our algorithm is
\[ m \cdot \exp(O((\log m)^{335/336} \log \log m)) \] where $336 = 2 \cdot 168$.

\section*{Acknowledgments}
We thank Li Chen, Rasmus Kyng, Simon Meierhans, and Maximilian Probst Gutenberg for several helpful discussions.

\bibliographystyle{alpha}
\bibliography{refs}

\end{document}